\documentclass[english]{article}
\usepackage[T1]{fontenc}
\usepackage[latin1]{inputenc}
\usepackage{geometry}
\usepackage{float}
\usepackage{mathrsfs}
\usepackage{amsmath}
\usepackage{amssymb}
\usepackage{graphicx}
\usepackage{subfigure}

\def\ou{Ornstein-Uhlenbeck\xspace}

\makeatletter

\floatstyle{ruled}
\newfloat{algorithm}{tbp}{loa}
\providecommand{\algorithmname}{Algorithm}
\floatname{algorithm}{\protect\algorithmname}

\usepackage{amsthm}

\usepackage{mathrsfs}

\usepackage{amsfonts}

\usepackage{epsfig}

\usepackage{bm}

\usepackage{mathrsfs}

\usepackage{enumerate}

\@ifundefined{definecolor}{\@ifundefined{definecolor}
 {\@ifundefined{definecolor}
 {\usepackage{color}}{}
}{}
}{}

\usepackage{subfig}\usepackage[all]{xy}

\newcommand{\bbR}{\mathbb R}

\newtheorem{theorem}{Theorem}[section]

\newtheorem{rem}{Remark}[section]
\newtheorem{prop}{Proposition}[section]

\newtheorem{ass}{Assumption}[section]
\newcounter{hypA}
\newenvironment{hypA}{\refstepcounter{hypA}\begin{itemize}
  \item[({\bf A\arabic{hypA}})]}{\end{itemize}}

\newcommand{\Exp}{\mathbb{E}}

\newcommand{\bbE}{\mathbb{E}}
\newcommand{\bbP}{\mathbb{P}}

\newcommand{\cA}{\mathcal{A}}

\newcommand{\cO}{\mathcal{O}}

\newcommand{\cB}{\mathcal{B}}

\newcommand{\cU}{\mathcal{U}}
\newcommand{\cV}{\mathcal{V}}

\newcommand{\heta}{\widehat{\eta}}

\newcommand{\cC}{\mathcal{C}}

\usepackage{xspace}
\usepackage{tabu}
\usepackage{booktabs}

\def\gbm{GBM\xspace}

\def\nlm{NLM\xspace}
\def\ou{OU\xspace}

\def\sde{SDE\xspace}
\def\bbR{\mathbb{R}}
\def\calL{\mathcal{L}}
\def\calN{\mathcal{N}}

\usepackage{babel}\date{}

\usepackage{babel}

\makeatother

\usepackage{babel}
\begin{document}

\begin{center}

{\Large \textbf{Multilevel Particle Filters: Normalizing Constant Estimation}}

\vspace{0.5cm}

BY AJAY JASRA$^{1}$, KENGO KAMATANI$^{2}$, PRINCE PEPRAH OSEI$^{1}$ \& YAN ZHOU$^{1}$

{\footnotesize $^{1}$Department of Statistics \& Applied Probability,
National University of Singapore, Singapore, 117546, SG.}
{\footnotesize E-Mail:\,}\texttt{\emph{\footnotesize staja@nus.edu.sg; op.peprah@u.nus.edu; stazhou@nus.edu.sg}}\\
{\footnotesize $^{2}$Graduate School of Engineering Science, Osaka University, Osaka, 565-0871, JP.}
{\footnotesize E-Mail:\,}\texttt{\emph{\footnotesize kamatani@sigmath.es.osaka-u.ac.jp}}
\end{center}

\begin{abstract}
In this article we introduce two new
estimates of the normalizing constant (or marginal likelihood)
for partially observed diffusion (POD) processes, with discrete observations. 
One estimate is biased but non-negative and the other is unbiased
but not almost surely non-negative.
Our method uses the
multilevel particle filter of \cite{ourmlpf}. We show that, under assumptions, for Euler discretized PODs
and a given $\varepsilon>0$
 in order to obtain a mean square error (MSE)
of $\mathcal{O}(\varepsilon^2)$ one requires a work of $\mathcal{O}(\varepsilon^{-2.5})$ 
for our new estimates
versus a standard particle filter that requires a work
of $\mathcal{O}(\varepsilon^{-3})$. Our theoretical results are supported by numerical simulations.
\end{abstract}

\section{Introduction}

We consider the filtering problem for partially observed diffusion processes, with discrete observations,
in particular we focus on the online estimation of the associated normalizing constant, or marginal likelihood
as data arrive. This class of problems has a wide number of applications, including finance, economics
and engineering; see for instance \cite{Cappe_2005}. Indeed, the marginal likelihood is a key quantity
for instance in model selection.

The framework will be described precisely in the next section, but essentially one considers
sequentially approximating probability measures $\{\heta_{n}\}_{n\geq 1}$ on a common  
space $\mathbb{R}^d$; assume that the probabilities have common dominating
$\sigma-$finite-measure $dx$. One is thus interested in computing, for many $\heta_n-$integrable
measurable, real-valued functions $\varphi$:
$$
\bbE_{\heta_{n}}[\varphi(X)] = \int_{\mathbb{R}^d} \varphi(x) \heta_n(x)dx\quad n\geq 1.
$$
For the specific problem of interest, for each $\{\heta_{n}\}_{n\geq 1}$ one can
only hope to sample from an approximation of which we assume that there $l=0,1,\dots$ of them.
That is, for $n$ fixed there is a sequence of approximations $\heta_{n}^0,\heta_{n}^1,\dots$,
where we assume that $\heta_{n}^0$ is a `poor' but fast (in some abstract computational sense) approximation of $\heta_{n}$
and as $l$ grows the approximations are more accurrate but slow. In our context the approximation is associated to time discretization
of the partially observed diffusion process.
In order to perform estimation for such models, even the approximations, one often has to resort to numerical methods such as particle filters; see e.g.~\cite{doucet_johan}.

The multilevel Monte Carlo (MLMC) framework  \cite{Giles08, giles_acta,heinrich2001multilevel}
allows one to leverage in an optimal way 
the nested problems arising in this context, hence minimizing the necessary cost to
obtain a given level of mean square error.  
In particular, the  multilevel Monte Carlo method seeks to sample from
$\heta_n^0$ as well as 
a sequence of coupled pairs $(\heta_n^1,\heta_n^2),\dots, (\heta_{n}^{L-1},\heta_n^{L})$
for some prespecified $L\geq 1$
and using a collapsing sum representation of $\bbE_{\heta_n^{L}}[\varphi(X)]$. Then, for some problems, using a 
suitable trade off of computational effort, one can reduce the amount of work, relative to i.i.d.~sampling 
from $\heta_n^{L}$ and using Monte Carlo integration, for a given amount of error.
However, we are concerned with the scenario where such independent sampling is not possible
and that one seeks to perform recursive estimation.

The former problem, that is the use of MLMC for recursive estimation was addressed in
\cite{ourmlpf}. That article developed a particle framework that explicitly uses the MLMC idea.
The authors showed that for the estimation of the filter, under assumptions and a particular context, the 
work to obtain a MSE of $\mathcal{O}(\varepsilon^2)$ is $\mathcal{O}(\varepsilon^{-2.5})$ versus
$\mathcal{O}(\varepsilon^{-3})$ for an ordinary particle filter approximating $\{\heta_{n}^L\}_{n\geq 1}$.
In this article we further extend the framework to consider the estimation of the normalizing constant.
We introduce two new estimators of the normalizing constant.
One estimate is biased but non-negative and the other is unbiased
but not almost surely non-negative.
We show that, under assumptions, in order to obtain a mean square error (MSE)
of $\mathcal{O}(\varepsilon^2)$ (for both new estimators) one requires a work of $\mathcal{O}(\varepsilon^{-2.5})$ versus a standard particle filter that requires a work
of $\mathcal{O}(\varepsilon^{-3})$. Our results do not consider the time parameter.

This paper is structured as follows.
In Section \ref{sec:setup} we give the model, filtering problem and the approximation of the model
that is considered. In Section \ref{sec:mlpf} the 
ML method is briefly reviewed, the MLPF is described and our new estimates given.
Section \ref{sec:theo} gives our theoretical results as well as a cost analysis of our
new estimator. Section \ref{sec:numerics} features numerical simulations. The appendix contains
technical proofs which support the conclusions in Section \ref{sec:theo}.

\section{Set Up}
\label{sec:setup}

\subsection{Model}

We consider the following partially-observed diffusion process:
\begin{eqnarray}
dX_t & = & a(X_t)dt + b(X_t)dW_t
\label{eq:sde}
\end{eqnarray}
with $X_t\in\mathbb{R}^d$, $t\geq 0$, $X_0$ given and $\{W_t\}_{t\in[0,T]}$ a Brownian motion of appropriate dimension. 
The following assumptions will be made on the diffusion process.
\begin{ass} The coefficients $a^j, b^{j,k} \in C^2$, for $j,k= 1,\ldots, d$. 
Also, $a$ and $b$ satisfy 
\begin{itemize}
\item[{\rm (i)}] {\bf uniform ellipticity}: $b(x)b(x)^T$ is uniformly positive definite;
\item[{\rm (ii)}] {\bf globally Lipschitz}:
there is a $C>0$ such that 
$|a(x)-a(y)|+|b(x)-b(y)| \leq C |x-y|$ 
for all $x,y \in \bbR^d$; 
\item[{\rm (iii)}] {\bf boundedness}: $\bbE |X_0|^p < \infty$ for all $p \geq 1.$
\end{itemize}
\label{ass:diff}
\end{ass}
Notice that (ii) and (iii) together imply that $\bbE |X_n|^p < \infty$ for all $n$.

It will be assumed that the data are 
regularly spaced (i.e.~in discrete time) observations 
$y_1,\dots,y_{n}$, $y_k \in \bbR^m$.
It is assumed that conditional on $X_{k\delta}$, for discretization $\delta>0$,
$Y_k$ is independent of all other random variables with density $G(x_{k\delta},y_k)$.
For simplicity of notation let $\delta=1$ (which can always be done by rescaling time), so $X_k = X_{k\delta}$.
The joint probability density of the observations and the unobserved diffusion at the observation times 
is then
$$
\prod_{i=1}^n G(x_{i},y_{i})Q^\infty(x_{(i-1)},x_{i}), 
$$
where 
$Q^\infty(x_{(i-1)},x)$ 
is the transition density of the diffusion process as a function of $x$, 
i.e. the density of the solution
$X_1$ of Eq. \eqref{eq:sde} at time $1$
given initial condition 
$X_0=x_{(i-1)}$.

%

For $k\in\{1,\dots,n\}$, the objective is to approximate the filter
$$
\heta_k^{\infty}(x_{k}|y_{1:k}) = \frac{\int_{\mathbb{R}^{(k-1)d}} \prod_{i=1}^k G(x_{i},y_{i})Q^\infty(x_{(i-1)},x_{i}) dx_{1:k-1}}{\int_{\mathbb{R}^{kd}} \prod_{i=1}^k G(x_{i},y_{i})Q^\infty(x_{(i-1)},x_{i}) dx_{1:k}}.
$$
Note that we will use $\heta_k^{\infty}$
as the notation for measure and density, with the use clear from the context.
It is also of interest,
as is the focus of this article, to estimate the normalizing constant,
or marginal likelihood 
$$
p_k^{\infty}(y_{1:k}) = \int_{\mathbb{R}^{kd}} \prod_{i=1}^k G(x_{i},y_{i})Q^\infty(x_{(i-1)},x_{i}) dx_{1:k}.
$$

\subsection{Approximation}

There are several issues associated to the approximation of the filter
and marginal likelihood, sequentially in time. Even if one knows
$Q^\infty$ pointwise, up-to a non-negative unbiased estimator, 
and/or can sample exactly from the associated law, 
advanced computational
methods, such as particle filters (e.g.~\cite{doucet_johan,fearn}), have to be adopted in order to estimate the filter.
  In the setting considered in this paper, it is assumed that one cannot
\begin{itemize}
\item{evaluate $Q^\infty$ pointwise, up-to a non-negative unbiased estimator}
\item{sample from the associated distribution of $Q^\infty$.}
\end{itemize}
$Q^\infty$ and its distribution must be approximated by some discrete time-stepping method \cite{KlPl92}.

It will be assumed that the diffusion process is 
approximated by a time-stepping method for time-step $h_l=2^{-l}$. 
For simplicity and illustration, Euler's method \cite{KlPl92} will be considered.
However, the results can easily be extended and the theory will be presented more generally.
In particular, 
\begin{eqnarray}
\label{eq:euler1step}
X^l_{k,(m+1)} & = & X^l_{k,m} + {h_l} a(X^l_{k,m}) + \sqrt{h_l} b(X^l_{k,m}) \xi_{k,m}, \\
\xi_{k,m} & \stackrel{\textrm{i.i.d.}}{\sim} & \mathcal{N}_d(0,I_d) 
\nonumber\end{eqnarray}
for $m=0,\dots, k_l$, where 
$k_l=2^l$ and $\mathcal{N}_d(0,I_d)$ is the $d-$dimensional normal distribution with mean zero and covariance the identity (when $d=1$ we omit the subscript). 
The numerical scheme gives rise to its own transition density between observation times 
$Q^l(x_{(k-1)},x )$, which is the density of $X^l_{(k-1),k_l}=X^l_{k,0}$ 
given initial condition $X^{l}_{(k-1),0}=x_{(k-1)}$.

\subsection{Monte Carlo Approximation at Time 1}

Suppose one aims to approximate the expectation of 
$\varphi \in \cB_b(\bbR^d)$ (the class of bounded, measurable and real-valued functions on $\mathbb{R}^d$).
Let $\heta_1^l(\varphi) := \bbE [\varphi(X_1^l)]$ for $l=0,\dots, \infty$.
That is, the expectation of $\varphi$ w.r.t.~multiple approximations of the fillter (levels) at time 1.
For a given $L$, if it were feasible, the
Monte Carlo approximation
of $\heta_1^\infty(\varphi)$ by 
$$
\heta_1^{L,N}(\varphi) = \frac{1}{N}\sum_{i=1}^N \varphi({X}^{L,i}_1), \qquad {X}^{L,i}_1 \stackrel{\textrm{i.i.d.}}{\sim} \heta_1^L\ ,
$$
has mean square error (MSE) given by 
 \begin{equation}
 \bbE | \heta_1^{L,N}(\varphi) -  \heta_1^\infty(\varphi) |^2 =   
  \underbrace{
\frac{\heta_1^L((\varphi-\heta_1^L(\varphi))^2)}{N}
}_{\rm variance} +
    \underbrace{| \heta_1^{L}(\varphi) -  \heta_1^\infty(\varphi) |^2}_{\rm bias}.
 \label{eq:mse}
 \end{equation}
 If one aims for $\cO(\varepsilon^2)$ MSE with optimal cost, then one must balance these two terms.  

For $l=0,1,\ldots, L$, the hierarchy of time-steps $\{h_l\}_{l=0}^L$ 
gives rise to a hierarchy of transition densities $\{Q^l\}_{l=0}^L$.  In 
some cases, it is well-known that the multilevel Monte Carlo (MLMC) method \cite{Giles08, heinrich2001multilevel} 
can reduce the cost to obtain a given level of mean-square error (MSE) \eqref{eq:mse}.
The description of this method and its extension to the particle filter setting of \cite{ourmlpf} will be the topic of the next 
section.

\section{Multilevel Particle Filters}
\label{sec:mlpf}

In this section, the multilevel particle filter will be discussed
and the contribution of this article, an unbiased ML estimator of the normalizing constant will be given.

\subsection{Multilevel Monte Carlo}
\label{ssec:mlmc}

The standard multilevel Monte Carlo (MLMC) framework \cite{Giles08} is described in the context
of approximating the filter at time 1. For $L\geq 1$ given, it is assumed for pedagogical purposes, that one can obtain samples
from $\heta_1^0$ and the couples $(\heta_1^0,\heta_1^1),\dots,(\heta_1^{L-1},\heta_1^{L})$, even though this is not possible in general.
The scenario described is somewhat abstract, but helps to understand the ML method in the context of the article.

The MLMC method begins with asymptotic estimates for 
weak and strong error rates, and the associated cost.  
In particular, assume that 
there are $\alpha, \beta, \gamma>0$, with $X_1^{\infty}\sim\heta_1^{\infty}$ and not necessarily independently $(X_{1,1}^l,X_{1,2}^l)\sim
(\heta_1^{l},\heta_1^{l-1})$
\begin{itemize}
\item[($B_{l}$)] $\bbE [ \varphi(X_{1,1}^{l}) - \varphi(X_{1}^{\infty}) ] = \cO(h_l^\alpha)$;
\item[($V_{l}$)] $\bbE [ |\varphi(X_{1,1}^{l}) - \varphi(X_{1,2}^{l})|^2 ] = \cO(h_l^\beta)$;
\item[($C_{l}$)] COST$(X_{1,1}^l,X_{1,2}^l) = \cO(h_l^{-\gamma})$,
\end{itemize}
where COST denotes the computational effort to obtain one sample $X_{1,1}^l,X_{1,2}^l$,   
and $h_l$ is the grid-size of the numerical method, 
for example as given in \eqref{eq:euler1step}.  In particular, for Euler method, $\alpha=\beta=\gamma=1$.
In general $\alpha\geq\beta/2$, as the choice $\alpha=\beta/2$ is always possible, by Jensen's inequality.

Recall that in order to minimize the effort to obtain a given MSE, one must balance the terms in 
\eqref{eq:mse}.
Based on ($B_{l}$) above, a bias error proportional to $\varepsilon$ will require 
\begin{equation}
\label{eq:ell}
L \propto 
-\log(\varepsilon)/\log(2)\alpha.
\end{equation}
Hence, the associated cost ($C_{L}$), in terms of $\varepsilon$,
for a given sample is $\cO(\varepsilon^{-\gamma})$.
Furthermore, the necessary number of samples to obtain a variance proportional to $\varepsilon^2$ for this 
standard single level estimator is given by $N \propto \varepsilon^{-2}$ following from standard calculations.
So the total cost to obtain a mean-square error tolerance of $\cO(\varepsilon^2)$ is: 
$\#$samples$\times($cost$/$sample)$=$total cost$\propto \varepsilon^{-2 - \gamma}$.
To anchor to the particular example of the Euler-Marayuma method, 
the total cost is $\cO(\varepsilon^{-3})$.

The idea of MLMC is to approximate the ML identity:
$$
\heta^L_1(\varphi) = \sum_{l=0}^{L-1} \{\heta^{l}_1(\varphi)-\heta^{l-1}_1(\varphi)\}
$$
with $\heta^{-1}_1(\varphi)\equiv 0$.
Let $N_0,\dots,N_L\in\mathbb{N}$ and
for $l\in\{0,\dots,L\}$, $(X_{1,1}^{l,i},X_{1,2}^{l,i})_{i=1}^{N_l}$ be i.i.d.~samples from the couple $(\heta_1^l,\heta_{l-1}^1)$
with the convention that for $i\in\{1,\dots,N_0\}$ $X^{0,i}_{1,2}$ is null. At this stage `couple' is abstract, in the sense that 
one wants  $(X_{1,1}^{l,i},X_{1,2}^{l,i})$ to be correlated such that condition ($V_{l}$) holds.

One can approximate the $l^{th}$ summand of the ML identity as
\begin{equation*}
Y_{l}^{N_l} (\varphi) := \frac{1}{N_l} \sum_{i=1}^{N_l}\{ \varphi(X_{1,1}^{l,i}) - \varphi(X_{1,2}^{l,i})\}.
\end{equation*}
The multilevel estimator is a telescopic sum of such unbiased increment estimators, 
which yields an unbiased estimator of $\eta^L_1(\varphi)$.  It can be  
defined in terms of its empirical measure as
\begin{equation*}
\heta_1^{L,{\rm Multi}} (\varphi) := \
\sum_{l=0}^L  Y_{l}^{N_l} (\varphi) \ , 
\end{equation*}
under the convention that $\varphi(X^{0,i}_{1,2})\equiv 0$.

The mean-square error of the multilevel estimator is given by
\begin{equation}
\begin{split}
\bbE \left \{ \heta_1^{L,{\rm Multi}} (\varphi) - \heta_1^\infty(\varphi) \right \}^2 & =  
\underbrace{ \sum_{l=0}^L \bbE \left \{ Y^{N_l}_{l}(\varphi) - [\heta_1^l(\varphi) - \heta_1^{l-1}(\varphi)] \right \}^2}_{\rm variance}  \\
&+ 
\underbrace{\{\heta_1^L(\varphi) - \heta_1^\infty(\varphi)\}^2}_{\rm bias}  .
\label{eq:mseml}
\end{split}
\end{equation}
The key observation is that the bias is given by the finest level, 
whilst the variance is decomposed into a sum of variances of the increments, which is of the form
$\cV = \sum_{l=0}^L V_l N_l^{-1}$.
By condition $(V_l)$ above, the variance of the $l^{th}$ increment has the form $V_l N_l^{-1}$ and $V_l = \cO(h_l^\beta)$.
 The total cost is given by the sum $\cC=\sum_{l=0}^L C_l N_l$.
Based on ($V_{l}$) and ($C_{l}$) above, 
optimizing $\cC$ for a fixed $\cV$ yields that $N_l = \lambda^{-1/2} 2^{-(\beta+\gamma)l/2}$, 
for Lagrange multiplier $\lambda$. 
In the Euler-Marayuma case $N_l = \lambda^{-1/2} 2^{-l}$.  
Now, one can see that after fixing the bias to $c\varepsilon$,
one aims to find the 
Lagrange multiplier $\lambda$ such that $\cV 
\approx c^2\varepsilon^2$.  
Defining $N_0=\lambda^{-1/2}$, then $\cV 
= N_0^{-1} \sum_{l=0}^L 2^{(\gamma-\beta)l/2}$, 
so one must have $N_0 \propto \varepsilon^{-2} K(\varepsilon)$, where $K(\varepsilon)=\sum_{l=0}^L 2^{(\gamma-\beta)l/2}$,
and the $\varepsilon$-dependence comes from $L(\varepsilon)$, as defined in \eqref{eq:ell}.
There are three cases, with associated $K$, and hence cost $\cC$, given in Table \ref{tab:mlcases}.

\begin{table}[h]
\begin{center}
  \begin{tabular}{ | c || c | c |}
    \hline
    CASE & $K(\varepsilon)$ & $\cC(\varepsilon)$ \\ \hline\hline
    $\beta>\gamma$ & $\cO(1)$ & $\cO(\varepsilon^{-2})$ \\ \hline
    $\beta=\gamma$  & $\cO(-\log(\varepsilon))$ & $\cO(\varepsilon^{-2}\log(\varepsilon)^2)$ \\ \hline
    $\beta<\gamma$  & $\cO(\varepsilon^{(\beta-\gamma)/2})$ & $\cO(\varepsilon^{-2+(\beta-\gamma)})$ \\
    \hline
  \end{tabular}
\end{center}
\caption{The three cases of multilevel Monte Carlo, and associated constant $K(\varepsilon)$ and cost $\cC(\varepsilon)$.}
\label{tab:mlcases}
\end{table}

For example, Euler-Marayuma falls into the case 
($\beta=\gamma$), so that $\cC(\varepsilon)= \cO(\varepsilon^{-2}\log(\varepsilon)^2)$.  In this case, one chooses 
$N_0 = C \varepsilon^{-2} |\log(\varepsilon)| = C 2^{2L} L$, where the purpose of $C$ is 
to match the variance with the bias$^2$, similar to the single level case.


\subsection{Multilevel Particle Filters}
\label{ssec:mlpf}

We now describe the MLPF method of \cite{ourmlpf}. 
The idea will be to run $L+1$ independent coupled particle filters which sequentally target
$\heta_k^0$ and the couples $(\heta_k^0,\heta_k^1),\dots,(\heta_k^{L-1},\heta_k^{L})$
for $k=1,2,\dots$. Each coupled particle filter will be run with $N_l$ samples. The algorithm
approximating $\heta_k^0$ is just the standard particle filter.

\subsection{Coupled Kernel}

In order to describe the MLPF, we need some definitions.
Let $l\geq 0$ be given, associated to the Euler discretization. Define a kernel, 
$M^l: [\bbR^d\times \bbR^d] \times [\sigma(\bbR^d)\times \sigma(\bbR^d)] \rightarrow \bbR_+$, 
where $\sigma(\cdot)$ denotes the $\sigma-$algebra of measurable subsets,
such that $M_{1}^l (x,A) := M^l([x,x'],A\times\bbR^d) 
\int_{A\times\bbR^d} M^l([x,x'],d[y,y']) = \int_A Q^l(x,dy)
= Q^l(x,A)$
and $M_2^l (x',A) := M^l([x,x'],\bbR^d \times A) = Q^{l-1}(x',A)$.

The kernel $M^l$ can be constructed using the following strategy.
First the finer discretization is simulated using \eqref{eq:euler1step} (ignoring the index $k$) with
$X^{l,i}_{0,1}=x_0$, 
for $i\in\{1,\dots, N_l\}$.
Now for the coarse discretization, let 
$X^{l,i}_{0,2}=x_0$ for $i\in\{1,\dots,N_l\}$, let $h_{l-1}=2h_l$ and 
for $m\in\{1,\dots,k_{l-1}\}$ simulate
\begin{equation}
X^{l,i}_{m+1,2} =  X^{l,i}_{m,2} + h_{l-1} a(X^{l,i}_{m,2}) + \sqrt{h_{l-1}} b(X^{l,i}_{m,2}) (\xi^i_{2m}+\xi^i_{2m+1}),  
\label{eq:euler1stepcoarse}
\end{equation}
where $\{\xi^i_{m}\}_{i=1,m=0}^{N_l,k_l}$ are the $i^{th}$ realizations used in the 
simulation of the finer discretization. This procedure will be used below.

\subsection{Identity and Algorithm}

Let $\varphi \in \cB_b(\bbR^d)$ and consider the following decomposition
\begin{eqnarray}
\heta^{\infty}_k(\varphi) & =  & \sum_{l=0}^L (\heta^{l}_k - \heta^{l-1}_k)(\varphi)  +  (\heta^{\infty}_k - \heta^{L}_k)(\varphi)
\label{eq:collapse_sum}
\end{eqnarray}
where  $\eta_m^{-1}(\varphi) :=0$. We will approximate the summands.

The {\bf multilevel particle filter (MLPF)} is given below:

{\bf For} $l=0,1,\dots,L$ and $i=1,\dots, N_l$, draw $(X^{l,i}_{1,1},X_{1,2}^{l,i})\stackrel{\textrm{i.i.d.}}{\sim} M^l((x_0,x_0),\cdot)$
(with the convention that $X_{1,2}^{0,i}$ is null for each $i$ and obvious extention for $M^0$).

{\bf Initialize} $k=1$.  {\bf Do}

\begin{itemize}
\item[(i)] {\bf For} $l=0,1,\dots,L$ and $i=1,\dots, N_l$, draw $(I_{k,1}^{l,i},I_{k,2}^{l,i})$ according to the coupled resampling procedure below. Set
$k = k+1$.
\item[(ii)] {\bf For} $l=0,1,\dots,L$ and $i=1,\dots, N_l$, independently draw $(X^{l,i}_{k,1},X^{l,i}_{k,2})| (x_{k-1,1}^{l,I_{k,1}^{l,i}},
x_{k-1,2}^{l,I_{k,2}^{l,i}})\sim M^l((x_{k-1,1}^{l,I_{k,1}^{l,i}},
x_{k-1,2}^{l,I_{k,2}^{l,i}}),~\cdot~)$;
\end{itemize}

The coupled resampling procedure for
the indices $I^{l,i}_{k,j}$, $j\in\{1,2\}$, is described below. To understand this, set:
\begin{equation}
w_{k,1}^{l,i} = \frac{G(x^{l,i}_{k,1},y_{k})}{\sum_{j=1}^{N_l} G(x^{l,j}_{k,1},y_{k})} \qquad {\rm and} \qquad
w_{k,2}^{l,i} = \frac{G(x^{l,i}_{k,2},y_{k})}{\sum_{j=1}^{N_l} G(x^{l,j}_{k,2},y_{k})}.
\label{eq:weights}
\end{equation}

\begin{itemize}
\item[{\bf a}.] with probability  $\alpha_k^l = \sum_{i=1}^{N_l}w_{k,1}^{l,i}\wedge w_{k,2}^{l,i}$, 
draw $I^{l,i}_{k,1}$ according to
$$
\bbP(I^{l,i}_{k,1}=j) = \frac{1}{\alpha_k^l} (w_{k,1}^{l,j}\wedge w_{k,2}^{l,j}),
\qquad j\in\{1,\ldots,N_l\}.
$$
and let $I^{l,i}_{k,2}=I^{l,i}_{k,1}$.
\item[{\bf b}.] otherwise, draw
$(I^{l,i}_{k,1}, I^{l,i}_{k,2})$ independently according to the probabilities 
\[
\begin{split}
\bbP(I^{l,i}_{k,1}=j) & =  [w_{k,1}^{l,j}-w_{k,1}^{l,j}\wedge w_{k,2}^{l,j}]/(\sum_{s=1}^{N_l}w_{k,1}^{l,s}-w_{k,1}^{l,s}  \wedge w_{k,2}^{l,s}); \\ 
\bbP(I^{l,i}_{k,2}=j) & =  [w_{k,2}^{l,j}-w_{k,1}^{l,j}\wedge w_{k,2}^{l,j}]/(\sum_{s=1}^{N_l}w_{k,2}^{l,s}-w_{k,1}^{l,s}\wedge w_{k,2}^{l,s}), 
\end{split}
\]
for $j\in\{1,\ldots,N_l\}$.
\end{itemize}

\cite{ourmlpf} show that each summand in the first term of \eqref{eq:collapse_sum} can be consistently estimated with:
$$
\sum_{i=1}^{N_l}\left \{ w_{k,1}^{l,i} \varphi(x^{l,i}_{k,1}) - 
w_{k,2}^{l,i} \varphi(x^{l,i}_{k,2}) \right \}.
$$
 In the Euler case, \cite{ourmlpf} show that under assumptions and not considering the time parameter,
how to choose $L$ and $N_{0:L}$ such that 
that for an MSE of $\mathcal{O}(\varepsilon^2)$ the work required is $\mathcal{O}(\varepsilon^{-2.5})$
whereas for a particle filter a MSE of $\mathcal{O}(\varepsilon^2)$ costs $\mathcal{O}(\varepsilon^{-3})$.

\subsection{Estimation of Normalizing Constants}

In the context of estimating $p^{l}(y_{1:k})$ for any fixed $l\geq 0$ and any $k\geq 1$, \cite{ourmlpf}
show that a non-negative unbiased estimator is
$$
\widehat{p}_1^{l,N_l}(y_{1:k}) = \prod_{j=1}^k \Big(\frac{1}{N_l}\sum_{i=1}^{N_l}G(x_{j,1}^{l,i},y_j)\Big).
$$
Note that for any $l\in\{1,\dots,L\}$, $p^{l-1}(y_{1:k})$ can unbiasedly be estimated by 
$$
\widehat{p}_2^{l-1,N_l}(y_{1:k}) = \prod_{j=1}^k \Big(\frac{1}{N_l}\sum_{i=1}^{N_l}G(x_{j,2}^{l,i},y_j)\Big).
$$
Clearly, these estimators do not take advantage of the ML principle. As
$$
p^L(y_{1:k}) = \sum_{l=0}^L\{
p^l(y_{1:k}) - p^{l-1}(y_{1:k})
\}
$$
with $p^{-1}(y_{1:k}) \equiv 0$, one can define the following ML unbiased estimator of $p^{L}(y_{1:k})$
$$
\widehat{p}^{L,N_{0:L}}(y_{1:k}) = 
\sum_{l=0}^L\{
\widehat{p}_1^{l,N_l}(y_{1:k}) - \widehat{p}_2^{l-1,N_l}(y_{1:k})
\}
$$
with $p_2^{-1,N_0}(y_{1:k})\equiv 0$. It is remarked that such an estimator is not almost surely non-negative.

We also consider the biased, but non-negative estimator:
$$
\widehat{p}_1^{0,N_0}(y_{1:k}) \prod_{l=1}^L \frac{\widehat{p}_1^{l,N_l}(y_{1:k})}{\widehat{p}_2^{l-1,N_l}(y_{1:k})}.
$$

\section{Theoretical Results}
\label{sec:theo}

The calculations leading to the results in this section are performed via a Feynman-Kac type representation 
(see \cite{delm:04,delmoral1}) which is detailed in the appendix.  

\subsection{Main Theorems}

For the main Theorem we assume the assumptions (A1-2) in the appendix, which are supposed
such that they hold for each level. This assumption is termed (A). Below $\mathbb{E}$
denotes expectation w.r.t.~the stochastic process that generates the MLPF and 
$\overline{B}_l(n)$ is the (level-dependent) constant associated to
\eqref{eq:b_rec} (in the appendix). 

\begin{theorem}
Assume (A). Then for any $n,L\geq 1$, there exist $\{B_l(n)\}_{0\leq l \leq L}$ with $B_l(n)<+\infty$, such that
for any $N_0,\dots,N_L\geq 1$
$$
\mathbb{E}\Big[\Big(\widehat{p}^{L,N_{0:L}}(y_{1:n})-p^{L}(y_{1:n})\Big)^2\Big] \leq
\sum_{l=0}^L \frac{B_l(n)}{N_l}.
$$
\end{theorem}

\begin{proof}
The proof follows easily from Proposition \ref{prop:main_control} in the appendix along with the unbiased property of the estimators.
\end{proof}

We now consider the MSE in the Euler case; recall $h_l=2^{-l}$. We ignore the time parameter $n$.
In this scenario, the bias is $\mathcal{O}(h_L^{-1})$, for $\varepsilon>0$ one should set
$L=\log(\varepsilon^{-1})/\log(2))$ to make the bias squared $\mathcal{O}(\varepsilon^2)$.
Following the work of \cite{ourmlpf} and Remark \ref{rem:cost}, in the Euler case, we have that
$$
\sum_{l=0}^L \frac{B_l(n)}{N_l} \leq C(n)\sum_{l=0}^L \frac{h_l^{1/2}}{N_l}
$$
for some constant $C(n)$. In the Euler case the cost of the algorithm at a given time is $\sum_{l=0}^L N_lh_l^{-1}$,
so one can use a simple constraind optimization procedure: minimizing the cost for a given variance
of $\mathcal{O}(\varepsilon^2)$ to find the values of $N_{0:L}$. One should take
$N_l = C\varepsilon^{-2}h_{l}^{3/4}K_L$ where $K_L=\sum_{l=0}^L h_l^{-1/4}\approx \varepsilon^{-1/4}$
leading to a cost of $C\varepsilon^{-9/4}\sum_{l=0}^L h_l^{-1/4}=C\varepsilon^{-2.5}$.
That is, 
in order to obtain a mean square error (MSE)
of $\mathcal{O}(\varepsilon^2)$ one requires a work of $\mathcal{O}(\varepsilon^{-2.5})$ versus a standard particle filter that requires a work
of $\mathcal{O}(\varepsilon^{-3})$. Our results do not consider the time parameter. If one considers the relative variance, we conjecture that under assumptions (see \cite{cdg:11})
one would need to scale $N_l$ linearly with time; this is left for future work.

For the biased estimator one can combine Proposition \ref{theo:nc_bias_theo} in Appendix \ref{app:bias} with
the above discussion to deduce the same information: in
order to obtain a mean square error (MSE)
of $\mathcal{O}(\varepsilon^2)$ one requires a work of $\mathcal{O}(\varepsilon^{-2.5})$ versus a standard particle filter that requires a work
of $\mathcal{O}(\varepsilon^{-3})$.

\section{Numerical Examples}\label{sec:numerics}

\subsection{Model Settings}

We will illustrate the numerical performance of the MLPF algorithm with a few
examples of the diffusion processes considered in this paper. Recall that, the
diffusions take the form,
\begin{equation*}
  d X_t = a(X_t) dt + b(X_t) d W_t, \qquad X_0 = x_0
\end{equation*}
with $X_t\in\bbR^d$, $t\ge0$ and $\{W_t\}_{t\in[0,T]}$ a Brownian motion of
appropriate dimension. In addition, partial observations $\{y_1,\dots,y_n\}$
are available with $Y_k$ obtained at time $k\delta$, and $Y_k|X_{k\delta}$ has
a density function $G(y_k,x_{k\delta})$. 
Details of each example described below. A summary of settings can be found in
Table~\ref{tab:model}.

\paragraph{Ornstein-Uhlenbeck Process}

First, we consider the following \ou process,
\begin{gather*}
  d X_t = \theta(\mu - X_t) + \sigma d W_t, \\
  Y_k|X_{k\delta} \sim \calN(X_{k\delta}, \tau^2), \qquad \varphi(x) = x.
\end{gather*}
An analytical solution exists for this process and the
exact value of $\Exp[X_{k\delta}|y_{1:k}]$ can be computed using a Kalman
filter. The constants in the example are, $x_0 = 0$, $\delta = 0.5$, $\theta =
1$, $\mu = 0$, $\sigma = 0.5$, and $\tau^2 = 0.2$.

\paragraph{Geometric Brownian Motion}

Next we consider the \gbm process,
\begin{gather*}
  d X_t = \mu X_t + \sigma X_t d W_t, \\
  Y_k|X_{k\delta} \sim \calN(\log X_{k\delta}, \tau^2), \qquad
  \varphi(x) = x,
\end{gather*}
This process 
also admits an analytical solution, by using the transformation $Z_t = \log X_t$.
The constants are, $x_0 = 1$, $\delta = 0.001$, $\mu = 0.02$, $\sigma = 0.2$
and $\tau^2 = 0.01$.

\paragraph{Langevin Stochastic Differential Equation}

Here the \sde is given by 
\begin{gather*}
  d X_t = \frac{1}{2}\nabla\log\pi(X_t) + \sigma W_t, \\
  Y_k|X_{k\delta} \sim \calN(0, \tau^2e^{X_{k\delta}}),\qquad
  \varphi(x) = \tau^2e^x
\end{gather*}
where $\pi(x)$ denotes a probability density function. In this example, we
choose the Student's $t$-distribution with degrees of freedom $\nu = 10$. The
other constants are, $x_0 = 0$, $\delta = 1$, $\sigma = 1$ and $\tau^2 = 1$.

\paragraph{An \sde with a Non-Linear Diffusion Term}

Last, we consider the following \sde,
\begin{gather*}
  d X_t = \theta(\mu - X_t) + \frac{\sigma}{\sqrt{1 + X_t^2}} d W_t, \\
  Y_k|X_{k\delta} \sim \calL(X_{k\delta}, s), \qquad \varphi(x) = x,
\end{gather*}
where $\calL(m,s)$ denotes the Laplace distribution with location $m$ and scale
$s$. The constants are $x_0 = 0$, $\delta = 0.5$, $\theta = 1$, $\mu = 0$,
$\sigma = 1$ and $s = \sqrt{0.1}$. We will call this example $\nlm$ for short
in the remainder of this section. 

\begin{table}
  \begin{tabu}{X[2l]X[2l]X[l]X[2l]X[l]}
    \toprule
    Example & $a(x)$ & $b(x)$ & $G(y;x)$ & $\varphi(x)$ \\
    \midrule
    \ou
    & $\theta(\mu - x)$
    & $\sigma$
    & $\calN(x, \tau^2)$
    & $x$ \\
    \gbm
    & $\mu x$
    & $\sigma x$
    & $\calN(\log x, \tau^2)$~~~~
    & $x$ \\
    Langevin
    & $\frac{1}{2}\nabla\log\pi(x)$~~~
    & $\sigma$
    & $\calN(0, \tau^2e^x)$
    & $\tau^2 e^x$ \\
    \nlm
    & $\theta(\mu - x)$
    & $\frac{\sigma}{\sqrt{1 + x^2}}$
    & $\calL(x, s)$
    & $x$ \\
    \bottomrule
  \end{tabu}
  \caption{Model settings. $\mathcal{L}$ is used to denote a log-normal distribution.}
  \label{tab:model}
\end{table}

\subsection{Simulation Settings}

For each example, we consider estimates at level $L = 1,\dots,8$. For the OU
and GBM processes, the ground truth is computed through a Kalman filter. For
the two other examples, we use results from particle filters at level $L = 9$
as approximations to the ground truth. 

For each level of MLPF algorithm, $N_l = \lfloor N_{0,L} h_l^{(\beta
  + 2 \gamma) / 4} \rfloor$ particles are used, where $h_l = M_l^{-1} = 2^{-l}$
is the width of the Euler-Maruyama discretization; $\gamma$ is the rate of
computational cost, which is $1$ for the examples considered here; and $\beta$
is the rate of the strong error. The value of $\beta$ is $2$ if the
diffusion term $b(x)$ is constant and $1$ in general. The value
$N_{0,L}\propto\varepsilon^{-2}K(\varepsilon)$ is set to $2^{2L}L$ for the
cases where the diffusion term is constant and $2^{(9/4)L}$ otherwise.
Resampling is done adaptively. For the plain particle filters, resampling is
done when ESS is less than a quarter of the particle numbers. For the coupled
filters, we use the ESS of the coarse filter as the measurement of
discrepancy. Each simulation is repeated 100 times.  

\subsection{Results}

The magnitude of the normalizing constants 
typically
grows linearly with $n$ on logarithm
scales. Thus to make a sensible comparison of the variances at different
time points, we multiply the value of $p(y_{1:n})$ or its estimators by $c^n$,
where $c$ is a constant independent of the samples and data. In other words,
the variance and MSE results shown below are up to a multiplicative constant
which only depends on $n$.

We begin by considering the rate $\beta/2$ of the strong error. This rate can
be estimated either by the sample variance of $\hat\varphi_l(n) =
\hat{p}_1^{l,N_l}(y_{1:n}) - \hat{p}_2^{l-1,N_l}(y_{1:n})$, or by $1 - p_l(n)$,
where $p_l(n)$ is the probability of the coupled particles having the same
resampling index at time step $n$. The latter was examined in 
\cite{ourmlpf}
 and the results are identical since we are using
exactly the same model and simulation settings. In \cite{ourmlpf} the authors were
interested in the estimation of the (expectations w.r.t.~the) filter at specific time points, whilst here we are
interested in the normalizing constants. In Figure~\ref{fig:rate} we show the
estimated variance of $\hat\varphi_l(n)$ against $h_l$. The rates are
consistent with previous results. The estimated rates are about $1$ for the OU
and Langevin examples, and $0.5$ for the other two. In addition, the rates are
consistent for different time $n$.

\begin{figure}
 \includegraphics[width=\linewidth]{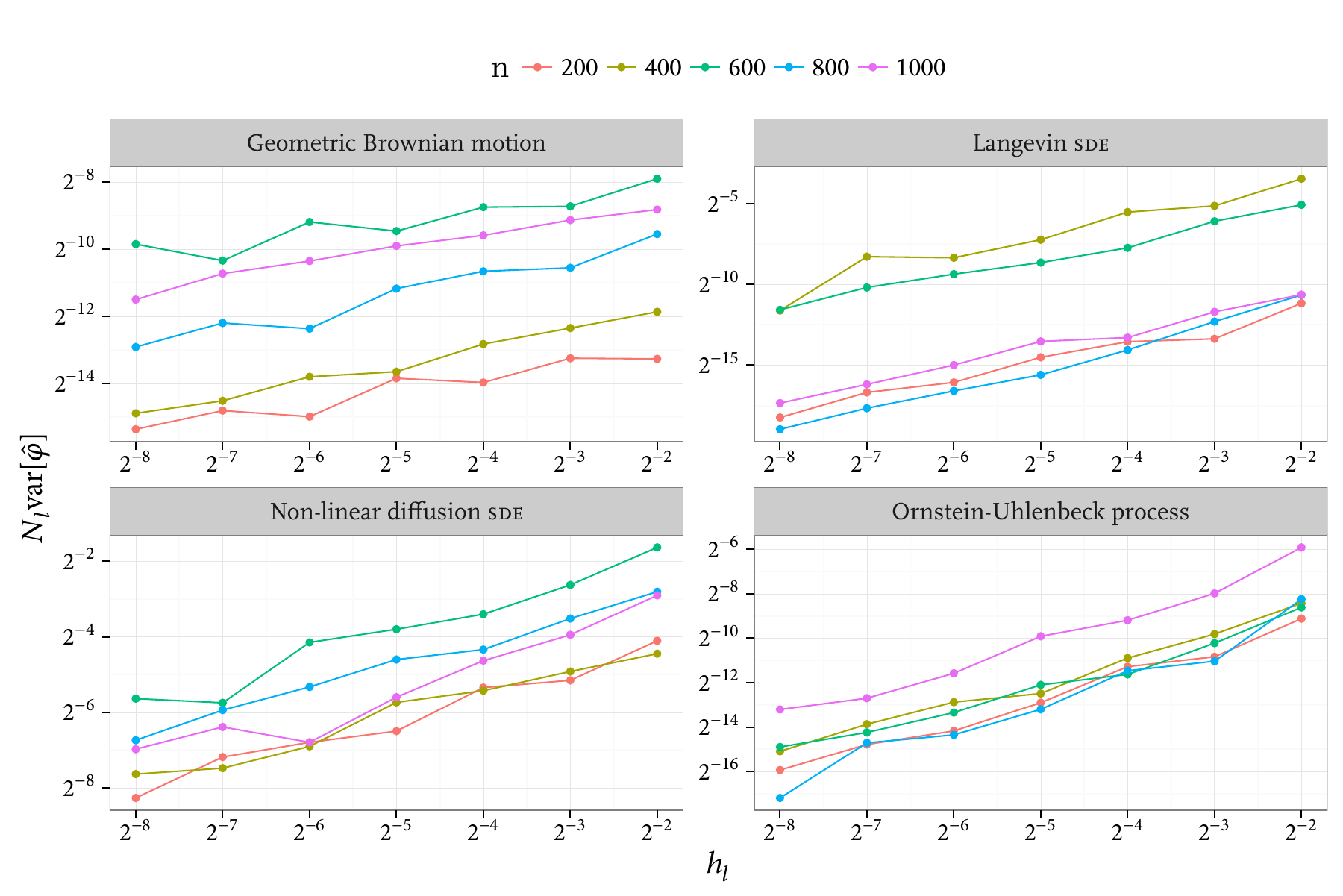}
  \caption{Rate estimates using the variance}
  \label{fig:rate}
\end{figure}

Next the rate of MSE vs.\ cost is examined. We consider the error of
normalizing constant at $n = 1000$. This is shown in Figure~\ref{fig:cost} and
Table~\ref{tab:cost}. Compared with results for estimates for test functions at
specific time points as in \cite{ourmlpf}, our
rates are slightly worse for both the PF and MLPF algorithms. However, they are
still consistent with the theory. Importantly, the MLPF algorithm,
using either the biased and unbiased estimators, shows significant advantage
over PF in all examples.

\begin{figure}
 \includegraphics[width=\linewidth]{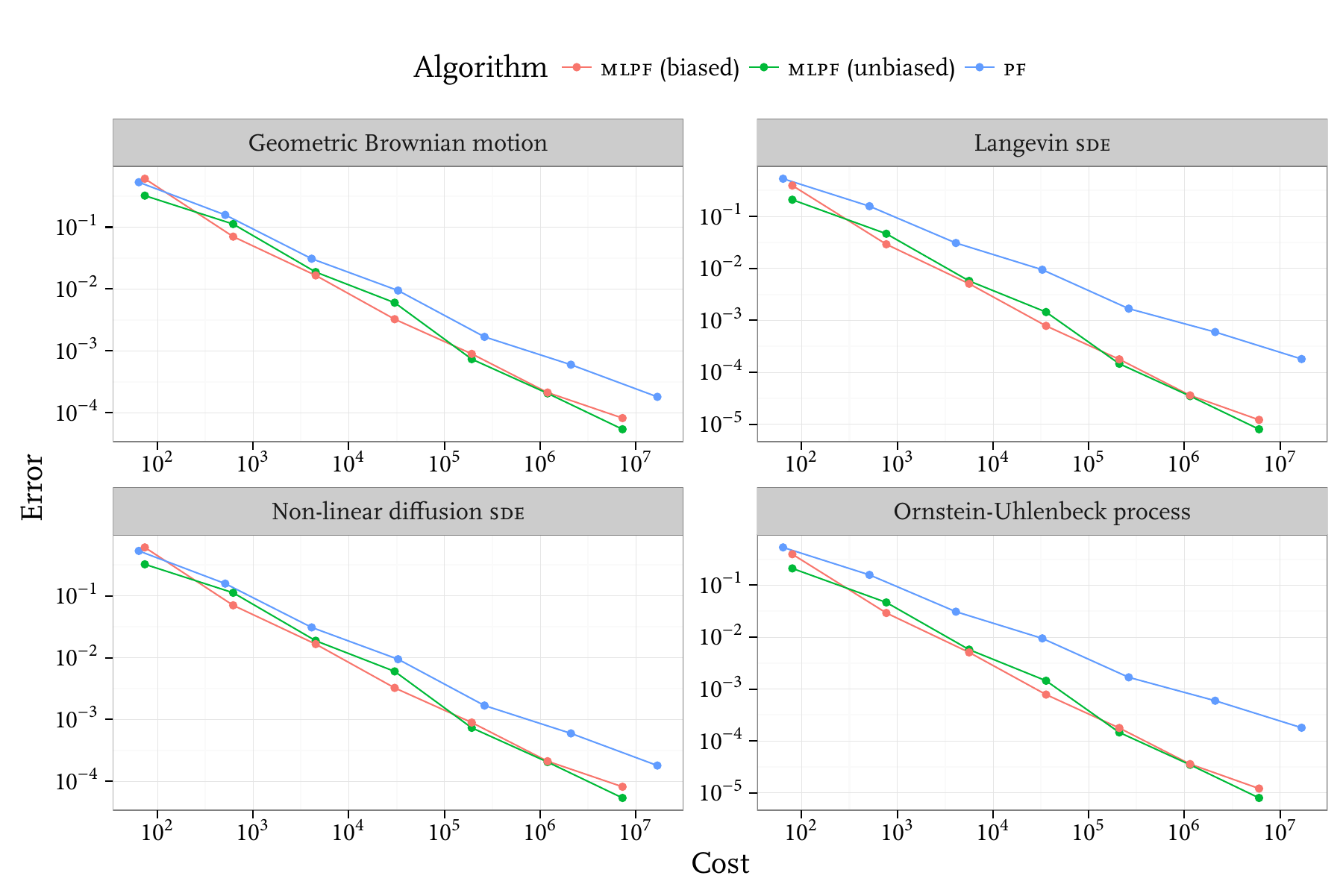}
  \caption{Cost rates}
  \label{fig:cost}
\end{figure}

\begin{table}
  \begin{tabu}{X[l]X[r]X[r]X[r]}
    \toprule
    Example & PF & MLPF (unbiased) & MLPF (biased) \\
    \midrule
    OU       & $-1.532$ & $-1.125$ & $-1.119$ \\
    GBM      & $-1.567$ & $-1.224$ & $-1.231$ \\
    Langevin & $-1.511$ & $-1.053$ & $-1.043$ \\
    NLM      & $-1.528$ & $-1.301$ & $-1.310$ \\
    \bottomrule
  \end{tabu}
  \caption{Cost rate $\log \text{cost} \sim \log\text{MSE}$}
  \label{tab:cost}
\end{table}

\subsubsection*{Acknowledgements}
AJ, PPO \& YZ were supported by an AcRF tier 2 grant: R-155-000-161-112. AJ is affiliated with the Risk Management Institute, the Center for Quantitative Finance 
and the OR \& Analytics cluster at NUS. KK \& AJ acknowledge CREST, JST for additionally supporting the research.

\appendix

\section{Set Up}

\subsection{Basic Notations}

 The 
total variation norm 
as $\|\cdot\|_{\textrm{tv}}$. The collection of real-valued Lipschitz functions on a space $E$ is written $\textrm{Lip}(E)$.
For two Markov kernels $M_1$ and $M_2$
on the same space $E$, letting $\cA=\{\varphi: \|\varphi\|\leq 1, \varphi\in\textrm{Lip}(E)\}$
write
$$
|||M_{1}-M_{2}||| := \sup_{\varphi\in \cA}\sup_x |\int_E \varphi(y) M_1(x,dy) - \int_E \varphi(y) M_2(x,dy) |.
$$

Consider a sequence of random variables $(v_n)_{n\geq 0}$ with $v_n=(u_{n,1},u_{n,2})\in \cU\times \cU =: \cV$. 
For $\mu\in\mathcal{P}(\cV)$ (the probability measures on $\cV$) and function $\varphi\in\mathcal{B}_b(\cU)$ (bounded-measurable, real-valued) we will write:
$$
\mu(\varphi_j) = \int_\cV \varphi(u_j) \mu(dv)\qquad j\in\{1,2\}.
$$
Write the $j\in\{1,2\}$ marginals (on $u_j$) of a probability $\mu\in\mathcal{P}(\cV)$ as $\mu_j$.
Define the potentials:
$G_n:\cU\rightarrow\mathbb{R}_+$. 
 Let $\eta_0\in\mathcal{P}(\cV)$ and define
Markov kernels $M_{n}:\cV\rightarrow \mathcal{P}(\cV)$ with $n\geq 1$. It is explictly assumed that for $\varphi\in\mathcal{B}_b(\cU)$ the $j$ marginals satisfy:
\begin{equation}
M_{n}(\varphi_j)(v) = \int_\cV \varphi(u_j') M_n(v,dv') =  \int_\cU \varphi(u_j') M_{n,j}(u_j,du_j').\label{eq:marginal_m}
\end{equation}
We adopt the definition for $(v,\tilde{v})=((u_1,u_2),(\tilde{u}_1,\tilde{u}_2))$ of a sequence of Markov kernels 
$(\bar{M}_n)_{n\geq 1}$, $\bar{M}_n:\cV\times \cV\rightarrow\mathcal{P}(\cV)$
$$
\bar{M}_n((v,\tilde{v}),dv') := M_n((u_1,\tilde{u}_2),dv').
$$
In the main text $\cU=\bbR^d$.

\subsection{Marginal Feynman-Kac Formula}

Given the above notations and defintions we define the $j-$marginal Feynman-Kac formulae:
$$
\gamma_{n,j}(du_n) = \int \prod_{p=0}^{n-1} G_p(u_p) \eta_{0,j}(du_0) \prod_{p=1}^n M_{p,j}(u_{p-1},du_p)
$$
with for $\varphi\in\mathcal{B}_b(\cU)$
$$
\eta_{n,j}(\varphi) = \frac{\gamma_{n,j}(\varphi)}{\gamma_{n,j}(1)}.
$$
One can also define the sequence of Bayes operators, for $\mu\in\mathcal{P}(\cU)$
$$
\Phi_{n,j}(\mu)(du) = \frac{\mu(G_{n-1}M_{n,j}(\cdot,du))}{\mu(G_{n-1})}\qquad n\geq 1.
$$
Recall that for $n\geq 1$, $\eta_{n,j} = \Phi_{n,j}(\eta_{n-1,j})$.

\subsection{Feynman-Kac Formulae for Multilevel Particle Filters}

For $\mu\in\mathcal{P}(\cV)$ define for $u\in \cU$, $v\in \cV$:
\begin{eqnarray*}
G_{n,j,\mu}(u) & = & \frac{G_{n}(u)}{\mu_j(G_{n})} \\
\bar{G}_{n,\mu}(v) & = & G_{n,1,\mu}(u_1) \wedge G_{n,2,\mu}(u_2).
\end{eqnarray*}

Now for any sequence $(\mu_n)_{n\geq 0}$, $\mu_n\in\mathcal{P}(\cV)$, define the sequence of operators $(\bar{\Phi}_n(\mu_{n-1}))_{n\geq 1}$:
$$
\bar{\Phi}_n(\mu_{n-1})(dv_n) = 
$$
$$
\mu_{n-1}(\bar{G}_{n-1,\mu_{n-1}})\frac{\mu_{n-1}(\bar{G}_{n-1,\mu_{n-1}}M_n(\cdot,dv_n))}{\mu_{n-1}(\bar{G}_{n-1,\mu_{n-1}})} +(1-\mu_{n-1}(\bar{G}_{n-1,\mu_{n-1}}))\times
$$
$$
\mu_{n-1}\otimes
\mu_{n-1}
\Big(
\Big[
\frac{G_{n-1,1,\mu_{n-1}}-\bar{G}_{n-1,\mu_{n-1}}}{\mu_{n-1}(G_{n-1,1,\mu_{n-1}}-\bar{G}_{n-1,\mu_{n-1}})}\otimes
\frac{G_{n-1,2,\mu_{n-1}}-\bar{G}_{n-1,\mu_{n-1}}}{\mu_{n-1}(G_{n-1,2,\mu_{n-1}}-\bar{G}_{n-1,\mu_{n-1}})}
\Big]
\bar{M}_n(\cdot,dv_n)
\Big)
$$
Now define $\bar{\eta}_n := \bar{\Phi}_n(\bar{\eta}_{n-1})$ for $n\geq 1$, $\bar{\eta}_0=\eta_0$. The following Proposition is proved in \cite{ourmlpf}:

\begin{prop}\label{prop:marginal}
Let $(\mu_n)_{n\geq 0}$ be a sequence of probability measures on $\cV$ with $\mu_0=\eta_0$ and for each $j\in\{1,2\}$, $\varphi\in\mathcal{B}_b(\cU)$
$$
\mu_{n}(\varphi_j) = \eta_{n,j}(\varphi).
$$
Then:
$$
\eta_{n,j}(\varphi) = \bar{\Phi}_n(\mu_{n-1})(\varphi_j).
$$
In particular $\bar{\eta}_{n,j}=\eta_{n,j}$ for each $n\geq 0$.
\end{prop}

%

The point of the proposition is that if one has a system that samples 
$\bar{\eta}_0$, $\bar{\Phi}_1(\bar{\eta}_0)$ and so on, that marginally, one has exactly the marginals
$\eta_{n,j}$ at each time point. In practice one cannot do this, but rather  samples at time $0$
$$
\Big(\prod_{i=1}^N \bar{\eta}_0(dv_0^i)\Big).
$$
Writing the empirical measure of the samples as $\bar{\eta}^N_{0}$
one then samples
$$
\prod_{i=1}^N \bar{\Phi}_{p}(\bar{\eta}^N_{p-1})(dv_p^i).
$$
Again writing the empirical measure as $\bar{\eta}^N_{1}$ and so on, one
runs the following system:
$$
\Big(\prod_{i=1}^N \bar{\eta}_0(dv_0^i)\Big)\Big(\prod_{p=1}^n \prod_{i=1}^N \bar{\Phi}_{p}(\bar{\eta}^N_{p-1})(dv_p^i)\Big)
$$
which is exactly one pair of particle filters at a given level of the MLPF.

$\eta_{n,1}$ (and its approximation) will represent the predictor at time $n$ for a `fine' level and
$\eta_{n,2}$ (and its approximation) will represent the predictor at time $n$ for a `coarse' level.
The time index here is shifted backwards one, relative to the main text 
and this whole section only considers one coupled particle filter. This is all that is required due to the independence of
the particle filters.

\section{Normalizing Constant: Unbiased Estimator}

Note the following
$$
\gamma_{n,j}^N(1) = \prod_{p=0}^{n-1} \bar{\eta}_{p,j}^N(G_p)  
$$
to estimate $\gamma_{n,j}(1)$ ($p(y_{1:n})$ in the main text; recall the subscript $j\in\{1,2\}$  has 1 as the fine, 2 the coarse). This estimate is unbiased as proved in \cite{ourmlpf}.
We will be considering the analysis of
$$
\gamma_{n,1}^N(1) - \gamma_{n,2}^N(1).
$$
In the assumptions below $G_n$ is exactly $G(x_n,y_n)$ in the main text. $M_{n,1}$  (resp.~$M_{n,1}$) is simply
the finer (resp.~coarser) Euler discretized Markov transition (there is no time parameter for the transition kernel in the
main text).

\begin{hypA}
\label{hyp:A}
There exist $c>1$ and $C>0$, such that for all $n\geq 0$, $x,x' \in \mathcal{U}$
\begin{itemize}
\item[{\rm (i)}] {\bf boundedness}: $c^{-1} < G_n(x) < c$; 
\item[{\rm (ii)}] {\bf globally Lipschitz}: $|G_n(x) - G_n(x')| \leq C |x-x'|$.
\end{itemize}
\end{hypA}

\begin{hypA}
\label{hyp:B}
There exists a 
$C>0$ such that for each $u,u'\in \mathcal{U}$, $j\in\{1,2\}$ and $\varphi\in\mathcal{B}_b(\mathcal{U})\cap\textrm{Lip}(\mathcal{U})$
$$
|M_{n,j}(\varphi)(u) - M_{n,j}(\varphi)(u')| \leq C_n\|\varphi\|~|u-u'|.
$$
\end{hypA}

Let
\begin{eqnarray}
\label{eq:beea}
B(n) & = & \Big(\sum_{p=0}^n \{\mathbb{E}[\{|U_{p,1}^1-U_{p,2}^1|\wedge 1 \}
^{2}]^{1/2} +\|\eta_{p,1}-\eta_{p,2}\|_{\textrm{tv}}\}
\\ & &
+ \sum_{p=1}^n|||M_{p,1}-M_{p,2}|||\Big)^2
\nonumber
\end{eqnarray}
where $\mathbb{E}$ is expectation w.r.t.~the law associated to the algorithm described in this appendix.
Let $\overline{B}(0) = C B(0)$ and for $n\geq 1$
\begin{equation}
\overline{B}(n) = C(n)[B(n-1) + \overline{B}(n-1) + \|\eta_{n-1,1}-\eta_{n-1,2}\|_{\textrm{tv}}^2 + (\gamma_{n-1,1}(1)-\gamma_{n-1,2}(1))^2]\label{eq:b_rec}
\end{equation}
where $C(n)$ is a constant depending upon $n$.

\begin{prop}\label{prop:main_control}
Assume (A\ref{hyp:A}-\ref{hyp:B}). Then for any $n\geq 1$, $N\geq 1$:
$$
\mathbb{E}[([\gamma_{n,1}^N(1) - \gamma_{n,2}^N(1)] - [\gamma_{n,1}(1) - \gamma_{n,2}(1)] )^2] \leq \frac{\overline{B}(n)}{N}.
$$
\end{prop}

\begin{proof}
Throughout $C(n)$ is a constant that depends on $n$ whose value may change from line to line.
We prove the result by induction on $n$. The case $n=1$ follows by \cite[Theorem C.2]{ourmlpf}, so we go immediately to the case of a general $n>1$ and
assuming the result at $n-1$. We have
$$
[\gamma_{n,1}^N(1) - \gamma_{n,2}^N(1)] - [\gamma_{n,1}(1) - \gamma_{n,2}(1)]  = 
$$
$$
\prod_{p=0}^{n-2}\eta_{p,1}^N(G_p)[\eta_{n-1,1}^N(G_{n-1})-\eta_{n-1,2}^N(G_{n-1})]  +
\eta_{n-1,2}^N(G_{n-1})[\prod_{p=0}^{n-2}\eta_{p,1}^N(G_p)-\prod_{p=0}^{n-2}\eta_{p,2}^N(G_p)]  
$$
$$
-
\prod_{p=0}^{n-2}\eta_{p,1}(G_p)[\eta_{n-1,1}(G_{n-1})-\eta_{n-1,2}(G_{n-1})] +
\eta_{n-1,2}(G_{n-1})[\prod_{p=0}^{n-2}\eta_{p,1}(G_p)-\prod_{p=0}^{n-2}\eta_{p,2}(G_p)] 
$$
\begin{equation}
 =  T_1^N + T_2^N - (T_1 - T_2).\label{eq:main_decomp}
\end{equation}
By the $C_2-$inequality we can consider bounding $\mathbb{E}[(T_1^N-T_1)^2]$ and $\mathbb{E}[(T_2^N-T_2)^2]$ respectively in \eqref{eq:main_decomp}.

\noindent\textbf{Term} $\mathbb{E}[(T_1^N-T_1)^2]$.\\
We have
$$
\mathbb{E}[(T_1^N-T_1)^2] \leq 
$$
$$
2\mathbb{E}[(\gamma_{n-2,1}^N(1)[(\eta_{n-1,1}^N(G_{n-1})-\eta_{n-1,2}^N(G_{n-1}))-(\eta_{n-1,1}(G_{n-1})-\eta_{n-1,2}(G_{n-1}))])^2] +
$$
$$
2(\eta_{n-1,1}(G_{n-1})-\eta_{n-1,2}(G_{n-1}))^2\mathbb{E}[(\gamma_{n-2,1}^N(1)-\gamma_{n-2,1}(1))^2].
$$
The almost sure-boundedness of $\gamma_{n-2,1}^N(1)$ and \cite[Theorem C.2]{ourmlpf} means that
$$
\mathbb{E}[(\gamma_{n-2,1}^N(1)[(\eta_{n-1,1}^N(G_{n-1})-\eta_{n-1,2}^N(G_{n-1}))-(\eta_{n-1,1}(G_{n-1})-\eta_{n-1,2}(G_{n-1}))])^2] \leq C(n) \frac{B(n-1)}{N}.
$$
Proposition \ref{prop:tech_res} along with (A\ref{hyp:A}) gives
$$
(\eta_{n-1,1}(G_{n-1})-\eta_{n-1,2}(G_{n-1}))^2\mathbb{E}[(\gamma_{n-2,1}^N(1)-\gamma_{n-2,1}(1))^2] \leq \|\eta_{n-1,1}-\eta_{n-1,2}\|_{\textrm{tv}}^2
\frac{C(n)}{N}.
$$
Hence
$$
\mathbb{E}[(T_1^N-T_1)^2] \leq C(n) \Big[\frac{B(n-1)}{N} + \|\eta_{n-1,1}-\eta_{n-1,2}\|_{\textrm{tv}}^2\frac{1}{N}\Big].
$$

\noindent\textbf{Term} $\mathbb{E}[(T_2^N-T_2)^2]$.\\
We have
$$
\mathbb{E}[(T_2^N-T_2)^2] \leq 
$$
$$
2\mathbb{E}[\eta_{n-1,2}^N(G_{n-1})^2([\gamma_{n-1,1}^N(1) - \gamma_{n-1,2}^N(1)] - [\gamma_{n-1,1}(1) - \gamma_{n-1,2}(1)] )^2] +
$$
$$
2[\gamma_{n-1,1}(1) - \gamma_{n-1,2}(1)]^2\mathbb{E}[(\eta_{n-1,2}^N(G_{n-1})-\eta_{n-1,2}(G_{n-1}))^2].
$$
By (A\ref{hyp:A}) and the induction hypothesis
$$
\mathbb{E}[\eta_{n-1,2}^N(G_{n-1})^2([\gamma_{n-1,1}^N(1) - \gamma_{n-1,2}^N(1)] - [\gamma_{n-1,1}(1) - \gamma_{n-1,2}(1)] )^2] \leq 
C(n) \frac{\overline{B}(n-1)}{N}.
$$
By \cite[Proposition C.1]{ourmlpf}
$$
[\gamma_{n-1,1}(1) - \gamma_{n-1,2}(1)]^2\mathbb{E}[(\eta_{n-1,2}^N(G_{n-1})-\eta_{n-1,2}(G_{n-1}))^2] \leq
[\gamma_{n-1,1}(1) - \gamma_{n-1,2}(1)]^2\frac{C(n)}{N}.
$$
Hence
$$
\mathbb{E}[(T_2^N-T_2)^2] \leq  C(n)\Big[\frac{\overline{B}(n-1)}{N}+[\gamma_{n-1,1}(1) - \gamma_{n-1,2}(1)]^2\frac{1}{N}\Big].
$$
From here one can conclude the proof.
\end{proof}

\begin{prop}\label{prop:tech_res}
Assume (A\ref{hyp:A}-\ref{hyp:B}). Then for any $n\geq 1$ there exist a $C(n)<+\infty$ such that for any $N\geq 1$, $j\in\{1,2\}$
$$
\mathbb{E}[(\gamma_{n,j}^N(1) - \gamma_{n,j}(1))^2] \leq \frac{C(n)}{N}.
$$
\end{prop}

\begin{proof}
We prove the result by induction on $n$. The case $n=1$ follows by \cite[Proposition C.1]{ourmlpf}, so we go immediately to the case of a general $n>1$ and
assuming the result at $n-1$. We have
$$
\gamma_{n,j}^N(1) - \gamma_{n,j}(1) = \prod_{p=0}^{n-2}\eta_{p,j}^N(G_p)[\eta_{n-1,j}^N(G_{n-1})-\eta_{n-1,j}(G_{n-1})]
+ \eta_{n-1,j}(G_{n-1})[\gamma_{n-1,j}^N(1) - \gamma_{n-1,j}(1)].
$$
Thus, by the $C_2-$inequality:
$$
\mathbb{E}[(\gamma_{n,j}^N(1) - \gamma_{n,j}(1))^2] \leq 
$$
$$
2\mathbb{E}\Big[
\Big(\prod_{p=0}^{n-2}\eta_{p,j}^N(G_p)[\eta_{n-1,j}^N(G_{n-1})-\eta_{n-1,j}(G_{n-1})]\Big)^2
\Big] + 2\mathbb{E}\Big[
\Big(
\eta_{n-1,j}(G_{n-1})[\gamma_{n-1,j}^N(1) - \gamma_{n-1,j}(1)]
\Big)^2
\Big].
$$
Using the boundedness of the $\{G_p\}_{p\geq 0}$ and \cite[Proposition C.1]{ourmlpf} deals with the first term on the R.H.S.~of the inequality and
the induction hypothesis the second term.
\end{proof}

For the following result, it is assumed that $M_{n,1}$ and $M_{n,2}$ are induced by an Euler approximation and the discretization levels are $h/2$ and $h$. 

\begin{prop}\label{prop:euler}
Assume (A\ref{hyp:A}(i)). Then for any $n\geq 1$ there exist a $C(n)<+\infty$ such that for any $\varphi\in\mathcal{B}_b(\mathcal{U})$
$$
|\gamma_{n,1}(\varphi)-\gamma_{n,2}(\varphi)| \leq C(n) \sup_{u\in\mathcal{U}}|\varphi(u)|h.
$$
\end{prop}

\begin{proof}
We prove the result by induction on $n$. The case $n=1$ follows by \cite[eq.~2.4]{delm}, so we go immediately to the case of a general $n>1$ and
assuming the result at $n-1$. We have
$$
\gamma_{n,1}(\varphi)-\gamma_{n,2}(\varphi) = \gamma_{n-1,1}(G_{n-1})[\eta_{n,1}(\varphi)- \eta_{n,2}(\varphi)] + \eta_{n,2}(\varphi)[\gamma_{n-1,1}(G_{n-1})-\gamma_{n-1,2}(G_{n-1})].
$$
By \cite[Lemma D.1.]{ourmlpf} (assumption 4.2(i) of that paper holds for an Euler approximation) 
\begin{eqnarray*}
|\gamma_{n-1,1}(G_{n-1})[\eta_{n,1}(\varphi)- \eta_{n,2}(\varphi)]| & \leq & 2\sup_{u\in\mathcal{U}}|\varphi(u)|\gamma_{n-1,1}(G_{n-1})\|\eta_{n,1}(\varphi)- \eta_{n,2}\|_{\textrm{tv}} \\
& \leq & 2\sup_{u\in\mathcal{U}}|\varphi(u)|\gamma_{n-1,1}(G_{n-1}) h.
\end{eqnarray*}
The induction hypothesis yields 
$$
|\eta_{n,2}(\varphi)[\gamma_{n-1,1}(G_{n-1})-\gamma_{n-1,2}(G_{n-1})]| \leq \sup_{u\in\mathcal{U}}|\varphi(u)|C(n-1) \sup_{u\in\mathcal{U}}|G_{n-1}i(u)|h.
$$
The proof can then easily be completed.
\end{proof}

\begin{rem}\label{rem:cost}
In the Euler case, Proposition \ref{prop:euler} 
along with \cite[Lemma D.1.]{ourmlpf} and that $B(n)= \mathcal{O}(h^{1/2})$ (see \cite[Corollary D.1]{ourmlpf})
establishes that $\overline{B}(n) = \mathcal{O}(h^{1/2})$.
\end{rem}

\section{Normalizing Constant: Biased Estimator}\label{app:bias}

In order to follow this section, one must have read the previous sections of the appendix.
We now consider the case of the biased estimator. In this scenario, the full algorithm is
considered; that is, a single particle filter and $L$ coupled (but independent) particle filters. Let $n\geq 1$
be given.
We define $\gamma_{n,j}^l(1), j\in\{1,2\}$ as the normalizing constants associated
to level $l\in\{1,\dots,L\}$. We write $\gamma_{n,1}^1(1)$ as the normalizing constant
at the coarsest level. We set
$$
\gamma_{n,j}^{N_l}(1) = \prod_{p=0}^{n-1}\eta_{p,j}^{N_l}(G_p)
$$
with $j\in\{1,2\}$, $l\in\{1,\dots,L\}$, with an obvious extension to $\gamma_{n,1}^{N_0}(1)$.
We are to analyze the estimate:
$$
\gamma_{n,1}^{N_0}(1) \prod_{l=1}^L \frac{\gamma_{n,1}^{N_l}(1)}{\gamma_{n,2}^{N_l}(1)}.
$$
We denote by (A) that the assumptions (A1-2) in the previous section uniformly at each level
(where applicable).
We write $\overline{B}_l(n)$ to denote the level specific version of $\overline{B}(n)$ in the previous section.

\begin{prop}\label{theo:nc_bias_theo}
Assume (A). Then for any $n\geq 1$ there exist a $C(n)<+\infty$ such that for
any $L\geq 1$, $N_{0:L}\geq 1$ we have
$$
\mathbb{E}\Big[\Big(
\gamma_{n,1}^{N_0}(1) \prod_{l=1}^L \frac{\gamma_{n,1}^{N_l}(1)}{\gamma_{n,2}^{N_l}(1)} -
\gamma_{n,1}^{0}(1) \prod_{l=1}^L \frac{\gamma_{n,1}^{l}(1)}{\gamma_{n,2}^{l}(1)}
\Big)^2\Big] \leq 
$$
$$
C(n)\Big(\frac{1}{\sqrt{N}_0}
+ \sum_{l=1}^L\Big(
\frac{\overline{B}_l(n)^{1/2}}{\sqrt{N}_l} +
\frac{|\gamma_{n,1}^{l}(1)-\gamma_{n,2}^{l}(1)|}{\sqrt{N}_l}
\Big)
\Big)^2.
$$
\end{prop}

\begin{proof}
Note that, 
$$
\gamma_{n,1}^{N_0}(1) \prod_{l=1}^L \frac{\gamma_{n,1}^{N_l}(1)}{\gamma_{n,2}^{N_l}(1)} -
\gamma_{n,1}^{0}(1) \prod_{l=1}^L \frac{\gamma_{n,1}^{l}(1)}{\gamma_{n,2}^{l}(1)}
=
$$
$$
(\gamma_{n,1}^{N_0}(1)-\gamma_{n,1}^{0}(1))\prod_{l=1}^L \frac{\gamma_{n,1}^{N_l}(1)}{\gamma_{n,2}^{N_l}(1)}
+
\sum_{l=1}^L\Bigg(
\gamma_{n,1}^{0}(1) \prod_{t=1}^{l-1} \frac{\gamma_{n,1}^{t}(1)}{\gamma_{n,2}^{t}(1)}
\Big(
\frac{\gamma_{n,1}^{N_l}(1)}{\gamma_{n,2}^{N_l}(1)} - 
\frac{\gamma_{n,1}^{l}(1)}{\gamma_{n,2}^{l}(1)}
\Big)
\prod_{s=l+1}^{L} 
\frac{\gamma_{n,1}^{N_s}(1)}{\gamma_{n,2}^{N_s}(1)}
\Bigg).
$$
So by Minkowski
$$
\mathbb{E}\Big[\Big(
\gamma_{n,1}^{N_0}(1) \prod_{l=1}^L \frac{\gamma_{n,1}^{N_l}(1)}{\gamma_{n,2}^{N_l}(1)} -
\gamma_{n,1}^{0}(1) \prod_{l=1}^L \frac{\gamma_{n,1}^{l}(1)}{\gamma_{n,2}^{l}(1)}
\Big)^2\Big] \leq 
$$
$$
\Bigg(\mathbb{E}[(\gamma_{n,1}^{N_0}(1)-\gamma_{n,1}^{0}(1))^2]^{1/2}
\prod_{l=1}^L \mathbb{E}\Big[\frac{\gamma_{n,1}^{N_l}(1)^2}{\gamma_{n,2}^{N_l}(1)^2}\Big]^{1/2} +
$$
\begin{equation}
\sum_{l=1}^L\Bigg(
\gamma_{n,1}^{0}(1) \prod_{t=1}^{l-1} \frac{\gamma_{n,1}^{t}(1)}{\gamma_{n,2}^{t}(1)}
\mathbb{E}\Big[
\Big(
\frac{\gamma_{n,1}^{N_l}(1)}{\gamma_{n,2}^{N_l}(1)} - 
\frac{\gamma_{n,1}^{l}(1)}{\gamma_{n,2}^{l}(1)}
\Big)^2
\Big]^{1/2}
\prod_{s=l+1}^{L} \mathbb{E}\Big[
\Big(\frac{\gamma_{n,1}^{N_s}(1)}{\gamma_{n,2}^{N_s}(1)}\Big)^2
\Big]^{1/2}\Bigg)
\Bigg)^2.\label{eq:main_bias}
\end{equation}
By standard results in SMC:
\begin{equation}
\label{eq:cont_ft}
\mathbb{E}[(\gamma_{n,1}^{N_0}(1)-\gamma_{n,1}^{0}(1))^2]^{1/2} \leq 
\frac{C(n)}{\sqrt{N_0}}.
\end{equation}
Now
$$
\mathbb{E}\Big[
\Big(
\frac{\gamma_{n,1}^{N_l}(1)}{\gamma_{n,2}^{N_l}(1)} - 
\frac{\gamma_{n,1}^{l}(1)}{\gamma_{n,2}^{l}(1)}
\Big)^2
\Big]^{1/2} = 
$$
$$
\mathbb{E}\Big[
\frac{1}{[\gamma_{n,2}^{N_l}(1)\gamma_{n,2}^{l}(1)]^2}
\Big(\gamma_{n,2}^{l}(1)[\gamma_{n,1}^{N_l}(1)-\gamma_{n,2}^{N_l}(1)
-(\gamma_{n,1}^{l}(1)-\gamma_{n,2}^{l}(1))
] +
$$
$$
(\gamma_{n,1}^{l}(1)-\gamma_{n,2}^{l}(1))
[\gamma_{n,2}^{l}(1)-\gamma_{n,2}^{N_l}(1)]
\Big)^2
\Big]^{1/2}.
$$
So we have by Minkowski and the bounded property of the $\{G_n\}_{n\geq 0}$:
$$
\mathbb{E}\Big[
\Big(
\frac{\gamma_{n,1}^{N_l}(1)}{\gamma_{n,2}^{N_l}(1)} - 
\frac{\gamma_{n,1}^{l}(1)}{\gamma_{n,2}^{l}(1)}
\Big)^2
\Big]^{1/2} \leq
$$
$$
C(n)\Big(
\mathbb{E}[[\gamma_{n,1}^{N_l}(1)-\gamma_{n,2}^{N_l}(1)
-(\gamma_{n,1}^{l}(1)-\gamma_{n,2}^{l}(1))
]^2]^{1/2} + |\gamma_{n,1}^{l}(1)-\gamma_{n,2}^{l}(1)|
\mathbb{E}[(\gamma_{n,2}^{l}(1)-\gamma_{n,2}^{N_l}(1))^2]^{1/2}
\Big).
$$
Applying Proposition \ref{prop:main_control} and \ref{prop:tech_res} to the two expectations we obtain
\begin{equation}
\mathbb{E}\Big[
\Big(
\frac{\gamma_{n,1}^{N_l}(1)}{\gamma_{n,2}^{N_l}(1)} - 
\frac{\gamma_{n,1}^{l}(1)}{\gamma_{n,2}^{l}(1)}
\Big)^2
\Big]^{1/2} \leq
C(n)\Big(
\frac{\overline{B}_l(n)^{1/2}}{\sqrt{N}_l} +
\frac{|\gamma_{n,1}^{l}(1)-\gamma_{n,2}^{l}(1)|}{\sqrt{N}_l}
\Big).
\label{eq:sum_cont_bias}
\end{equation}
Conbining \eqref{eq:main_bias} with \eqref{eq:cont_ft} and \eqref{eq:sum_cont_bias}
along with 
the bounded property of the $\{G_n\}_{n\geq 0}$
allows one to conclude
\end{proof}

\end{document}